\documentclass[12pt]{article}
\usepackage[dvips]{color}
\usepackage{amsmath,amsfonts,amsthm,amssymb,bbm,graphicx}
\usepackage{enumerate}
\usepackage{newlfont}

\newcommand{\eps}{\ensuremath{\epsilon}}
\newcommand{\ul}{\ensuremath{U_{\lambda}}}

\newcommand{\Ql}{\ensuremath{\mathcal{Q}_\lambda}}

\newcommand{\Vl}{\ensuremath{V_{\lambda}}}

\newcommand{\essinf}{\ensuremath{ \mathrm{ess}\,\inf}}
\newcommand{\esssup}{\ensuremath{ \mathrm{ess}\,\sup}}
\newcommand{\C}{\ensuremath{\mathcal{C}}}
\newcommand{\K}{\ensuremath{\mathcal{K}}}
\newcommand{\um}{\ensuremath{U_{\mu}}}
\newcommand{\Vm}{\ensuremath{V_{\mu}}}

\newtheoremstyle{myplain}
{}
{}
{\itshape}
{}        
{\scshape}
{}
{.5em}
{}

\newtheoremstyle{myremark}
{}
{}
{}
{}        
{\scshape}
{}
{.5em}
{}

\theoremstyle{myplain}
\newtheorem{lemma}{Lemma}[section]

\newtheorem{corollary}[lemma]{Corollary}
\newtheorem{definition}[lemma]{Definition}
\newtheorem{assumption}[lemma]{Assumption}
\newtheorem{theorem}[lemma]{Theorem}
\theoremstyle{remark}
\newtheorem{remark}[lemma]{Remark}
\newtheorem{example}[lemma]{Example}
\begin{document}
\title {The best gain-loss ratio is a poor performance measure}
\author {Sara Biagini \thanks{Corresponding author. University of Pisa,  Dipartimento di Economia e  Management, via Cosimo Ridolfi 10, 56100 Pisa, Italy. Email:  sara.biagini@ec.unipi.it, phone:+39050221649} \and Mustafa \c{C}. P{\i}nar \thanks{Bilkent University
Department of Industrial Engineering,
06800 Bilkent Ankara, Turkey. Email: mustafap@bilkent.edu.tr, phone: +903122901514}} \maketitle

\begin{abstract}
 The gain-loss ratio is known to enjoy  very good properties from a normative point of view.   As a confirmation,  we show that  the best market gain-loss ratio in the presence of a random endowment   is an acceptability index and we provide its dual representation for general probability spaces.\\
\indent However, the gain-loss ratio was designed for  finite $\Omega$, and works best in that case. For general $\Omega$  and in most continuous time   models, the best gain-loss is either infinite or fails to be attained. In addition, it displays  an odd behaviour   due to the scale invariance property, which does not seem desirable in this context.  Such  weaknesses  definitely prove that the (best) gain-loss is a \emph{poor} performance measure.\\

\noindent {\bf Key words:} Gain-loss ratio, acceptability indexes, incomplete markets, martingales, quasi concave optimization, duality methods,  market modified risk measures.

\noindent{\bf JEL:} G11, G12, G13.  {\bf MSC 2010:} 46N10, 91G99, 60H99.
\\
\noindent \textbf{Acknowledgements} We warmly thank  Jocelyne Bion-Nadal, Ale\v{s} \v{C}ern\'{y}, Marco Frittelli and Paolo Guasoni for their   valuable suggestions. Special thanks go two anonymous referees for their careful reading and  remarks,  which  substantially   improved  the quality of the paper.
\end{abstract}

\section{Introduction}
The gain-loss ratio  was introduced by Bernardo and Ledoit \cite{bl} to provide  an alternative  to the classic  Sharpe Ratio (SR) in portfolio performance evaluation.  Cochrane and Saa-Requejo  \cite{csr} call portfolios with  high SR 'good deals'. These opportunities should, informally speaking,  be regarded as  quasi-arbitrages and therefore should be ruled out.  Ruling out good deals, or equivalently   restricting   SR,  produces  in turn  restrictions on pricing kernels. Restricted pricing kernels are  desirable  since they provide   narrower lower and upper
price intervals for contingent claims in comparison to arbitrage-free
price intervals.
  This criterion is based on the assumption that a high SR is attractive, and a low SR is not.  The SR criterion  works well in a Gaussian returns context, but in general it does not  since it is  incompatible with no-arbitrage. In fact  a positive gain   with  finite first moment but infinite variance has zero SR,  but it is    very attractive as it is an arbitrage.
  The Sharpe Ratio (SR) has another   drawback:  it is not monotone, and
  thus  violates a basic axiom in theory of choice. To remedy the afore-mentioned shortcomings of the SR,
  Bernardo and Ledoit   proposed  as  performance measure  the gain-loss ratio:
$$ \alpha(X) = \frac{E[X^+]}{E[X^-]} $$
  where the expectation is taken under the historical probability measure $P$. The gain-loss ratio $\alpha$ is   well defined   on non-null payoffs $X$ as soon as  $X^+ $ or $X^-$  are integrable, it has an intuitive significance  and is  easy to compute. It also enjoys many  properties: monotonicity across $X$s; scale invariance, that is $\alpha(c  X ) = \alpha (X)$ for all $c >0$;   law invariance, as two payoffs with the same distribution have the same $\alpha$; and a classic continuity property (Fatou  property).     Restricted to portfolios with positive expectation, it becomes  a quasi concave map, consistent with second order stochastic dominance, as shown by  Cherny and Madan in \cite{cm}, and is thus an \emph{acceptability index} in their terminology.

Let  $\alpha^*$ denote the best  gain-loss ratio from  the  market, i.e. from the set  $\mathcal{X}$ of  non-trivial, \emph{discounted}, portfolio gains with finite first moment: $$ \alpha^*: = \sup_{X \in \mathcal{X}, X \neq 0} \alpha(X).$$
In case $P$ is already a pricing kernel,  $\alpha^*=1$ as $E[X]=E[X^+-X^-]=0$ for all gains.  This gives a flavor of
 the  main result  by Bernardo and Ledoit, which  is  the equivalence between
 \begin{itemize}
   \item[i)]  $\alpha^* <+\infty$,
   \item[ii)]  existence of   pricing kernels  with state price density $Z$ satisfying $ c\leq  Z\leq C$ for some constants $C,c>0$.
 \end{itemize}
 That is, restrictions on the best gain-loss ratio are equivalent to the existence of special, restricted pricing kernels bounded and bounded away from $0$.  Bernardo and Ledoit also prove a duality formula for $\alpha^*$,
 \[ \alpha^* = \min_{Z} \frac{\esssup Z}{\essinf Z} \]
 where $Z$ varies  over all the pricing kernels as in item ii) above.
 Though stated for a general probability space and in a biperiodal market model, Bernardo and Ledoit's derivation is correct only if $\Omega$ is finite. In fact,  what they actually show is $$ \alpha^*= \max_{X \in \mathcal{X}, X \neq 0} \alpha(X) =   \min_{Z} \frac{\esssup Z}{\essinf Z},$$
 i.e. that the  best ratio is always attained. This is true only if $\Omega$ is finite. \\
\indent  Against this background, the present paper develops an
analysis of the gain-loss ratio for general probability spaces. The rest of the paper is organized as follows. In Section 2   we  show the above equivalence  i) $\iff$ ii)   in  the presence of a  continuous time market for general $\Omega$. The duality technique employed here  extends  also  Pinar's treatment  \cite{p1,p2}.  The  assumptions made on the market model are quite general, as we do not require  the    underlyings process $S$ to be  neither a continuous diffusion, nor locally bounded. \\
\indent The duality formula for $\alpha^*$  is correctly reformulated as $\sup \cdots = \min \cdots $  in Theorem \ref{ftap-simple}, and    a  simple counterexample where the  supremum $\alpha^*$, though finite,  is not attained is provided in the Examples  Section 2.4. \\
 \indent
  In Section 2.3 pros and cons of the best gain-loss ratio are discussed.  While in discrete time models there is a full characterization of models with finite best gain-loss ratio, in continuous time the situation is hopeless. In most commonly used models, $\alpha^* = +\infty$ as any pricing kernel is unbounded as shown in details  for the Black Scholes model in Example \ref{bs-model}.   Finally, in Section 3 we analyze the best gain-loss ratio $\alpha^*(B)$  in the presence of a random endowment $B$.  In Section 3.1    $\alpha^*(B)$  is shown to be an acceptability index on integrable payoffs,  according to the definition given by Biagini and Bion-Nadal \cite{bb}. There we   briefly highlight   the difference between  the   notions  of acceptability index as given in \cite{cm} and  \cite{bb}, and we motivate the reason why the choice made by \cite{bb} is preferable here.   Then, in Section 3.2 we prove an extension of   Theorem \ref{ftap-simple} in the presence of $B$ and we provide a dual representation for $\alpha^*(B)$.   Section 3.3 concludes  by   pointing out other  gain-loss drawbacks when an endowment is present, which prove that the (best) gain-loss is a poor performance measure.

\section{The market best gain-loss $\alpha^*$ and its dual representation}

\subsection{The market  model}
Let $(\Omega, (\mathcal{F}_{t})_{t \in [0,T]}, P)$ be a continuous time stochastic basis satisfying the usual assumptions.  $S$ is an $\mathbb{R}^d$-valued semimartingale  on this basis and models the (discounted) time evolution of $d$  underlyings up to the finite horizon $T$.  A strategy $\xi$ is predictable, $S$-integrable process and     the stochastic integral $ \xi \cdot S$ is  the corresponding  gain process. Now, some integrability condition must be imposed on $S$ in order to ensure the presence of strategies $\xi$ with well defined gain-loss ratio. In some  cases in fact it may   happen   that every non-null terminal  gain $K = \xi\cdot S_T$  verifies $E[K^+]=E[K^-]=+\infty$, see  the Examples Section   for a simple one period model of such an extreme situation. \\
\indent The following is thus the integrability assumption on $S$ which holds throughout the paper.
\begin{assumption} \label{ass-S}
 Let $S^*_T = \sup_{t\leq T} |S_t|$ denote the maximal functional at   $T$. Then $S^*_T\in L^1(P)$.
\end{assumption}

Note that  $S^*_T$ coincides with the running maximum at the terminal date $T$ if $S$ is non-negative.   This assumption  is verified in many models used in practice:
\begin{itemize}
  \item if time is discrete, with finite horizon, or  equivalently: $S$ is a pure jump process  with jumps occurring only at fixed dates $t_1, \ldots, t_n$,  the assumption is equivalent to  $S_{t_i} \in L^1(P)$ for all $t_i$;

  \item if $S$ is a L\'{e}vy process, the assumption is equivalent to the integrability of $S_T$ only (or of $S_t$ at any fixed $0<t\leq T)$. This is a particular case  of a more general result on moments of  L\'{e}vy process, see reference \cite[Section 5.25]{sato}  (specifically Theorem 5.25.18).
\end{itemize}

 Therefore,  at least in normal market conditions Assumption \ref{ass-S} is quite reasonable.  From a strict  mathematical perspective it ensures that the gains processes are true (and not local) martingales under bounded  pricing kernels.   The  admissible strategies  we consider  are the linear  space
 $ \Xi =\{ \xi \mid \xi  \text{ is simple, predictable and bounded}  \}$, i.e.   those $\xi $ which may be written as $\sum_{i=1}^{n-1} H_i \mathbbm{1}_{  ] \tau_{i}, \tau_{i+1} ] }$ for some stopping times $0\leq \tau_1 <  \ldots < \tau_n \leq T $ with $H_i$ bounded and $ \mathcal{F}_{\tau_{i}}$-measurable. These  strategies   represent the set of  buy-and-hold strategies on $S$ over
finitely many trading dates.  The set of terminal admissible gains, which are replicable at zero cost via a simple strategy,  is  thus  the linear space
 $$\mathcal{K} = \{ K \mid K=\xi\cdot S_T \text{ for some  } \xi \in \Xi\}.  $$
 Thanks to Assumption \ref{ass-S},  $\mathcal{K}\subseteq  L^1(P)$.  Note   that  $ \xi = \mathbbm{1}_A \mathbbm{1}_{]s,t]} $  and its opposite $-\xi$  are in $\Xi$ for all $A \in \mathcal{F}_s$ and for all $0\leq s<t\leq T$, so that $K = \mathbbm{1}_A (S_t-S_s) $ and $-K$ are in $ \mathcal{K}$.\\
 \indent The  best gain-loss  in the above market is  then
  $$\alpha^* := \sup_{K \in \K, K\neq 0 } \alpha(K). $$
  The best gain-loss $\alpha^*$ is always greater or equal to $1$, and it is equal to $1$ if and only if $P$ is already a martingale measure for $S$. These facts can be easily  proved, using the linearity of $\K$ and the  above observation: $ \pm \mathbbm{1}_A \mathbbm{1}_{]s,t]}  \in \Xi$.

\subsection{No $\lambda$ gain-loss, its dual characterization and the duality formula for $\alpha^*$ }

The market best gain-loss $\alpha^*$ is  the value of a  non-standard  optimization problem.   In fact, the gain-loss ratio $\alpha$ is  not concave, and not even   quasi concave on $L^1(P)$. However,  when restricted to variables with non-negative  expectation  it becomes quasi-concave, as shown in detail by \cite{cm}. Since  the optimization can be restricted to  gains with non-negative expectations without loss of generality, in the end $\alpha^*$ can be seen as the optimal value of a quasi concave problem.\\
  \indent To characterize $\alpha^*$ and to link it to a no-arbitrage type result,  we rely on  a parametric family of  auxiliary  utility maximization problems   with piecewise linear  utility $\ul$:
$$ \ul(x) = x^+-\lambda x^-, \ \  {\lambda\geq 1}. $$
The convex conjugate of $\ul$, $\Vl(y) = \sup_{x} (\ul(x)-xy)$  is   the functional   indicator   of the interval $[1, \lambda]$:   $$\Vl (y) = \left \{\begin{array}{cc}
                             0 &  \text{ if }  1 \leq y\leq \lambda \\
                             +\infty & \text{ otherwise.}
                           \end{array}
          \right.  $$

  By mere definition of the conjugate, the   Fenchel inequality holds:
   \begin{equation}\label{fenchel}
    \ul(x) -xy \leq \Vl(y)\ \ \  \text{ for all } x, y \in \mathbb{R}.
    \end{equation}

\begin{definition}
Fix $\lambda \in [1, +\infty)$. Then the  set of probabilities  $\Ql$ which  have finite $\Vl$ entropy  is:
 $$\mathcal{Q}_\lambda := \{Q \;\mbox{probab.}, Q\ll P \; \mid \exists y >0,   E[\Vl(y\frac{dQ}{dP})]<+\infty \}.$$

\end{definition}

\begin{remark}
The set $\Ql$ is  not empty, as     $ \mathcal{Q}_1=\{P\}$ and   $ P \in \Ql $ for all $\lambda \geq 1$. It is also easy to check that $\Ql$ is convex and  the family  $ (\Ql)_{\lambda\geq 1}$ is non-decreasing in the parameter.  With the usual convention $\frac{c}{0} =+\infty $ for $c>0$,   $\mathcal{Q}_\lambda = \{Q \;\mbox{probab.}, Q\ll P \; \mid \frac{\esssup \frac{dQ}{dP}}{ \essinf \frac{dQ}{dP}  } \leq \lambda\} $.
\end{remark}
The  next definition is understood as follows.  The market is gain-loss free at a certain level  ${\lambda >1 }$ if not only there is no gain with $\alpha \geq \lambda$, but also  $\lambda$ cannot be approximated arbitrarily well with gains in $\K$.

\begin{definition}
For a given $\lambda \in (1, +\infty)$,   the market is    $\lambda$ gain-loss free  if  $\alpha^*<\lambda $.
\end{definition}

Theorem \ref{ftap-simple} below, first shown by Bernardo and Ledoit in a  two periods setup,    states   the equivalence between absence of    $\lambda$ gain-losses  and existence of a martingale measure whose density satisfies  precise bounds.\\
 \indent Some notation first. Let $\C= \{ X \in L^1 \mid X\leq K  \text{ for some } K\in  \mathcal{K} \}$  denote  the set (convex cone) of claims which are  super replicable at zero cost, and consider its polar set $\C^0 =\{ Z \in L^{\infty} \mid E[ZX]\leq 0 \text{ for all }  X\in \C\}$. As $\C\supseteq -L^1_+$, $\C^0\subseteq L^{\infty}_+$.
  $\C^0$ is a convex cone and thus not empty as $0\in \C^0$. \\
  \indent However, $\C^0$ may be trivially $\{0\}$, i.e. its basis $\C^0_1=\{ Z \in \C^0 \mid E[Z]=1 \}$ may be empty. This  may  happen  in common models  such as the Black Scholes model, see Remark \ref{bs-unbounded} and Example \ref{bs-model} for a discussion and more details.
  The basis  $\C_1^0$   however is important for   gain-loss analysis.
  The following  Lemma in fact proves that   $\C^0_1$ is the set of bounded  martingale probability densities, which in turn appear in the characterization of the  market best gain-loss in Theorem \ref{ftap-simple}.
  \begin{lemma} \label{mart}
   $Z\in \C_1^0$ if and only if it  is  a bounded  martingale density.
  \end{lemma}
  \begin{proof}
   If $Z\in\C_1^0 $, it  is bounded non-negative and integrates to $1$, so   it is a probability density of a $Q\ll P$. Moreover,  $\pm \mathbbm{1}_A(S_t-S_s) \in \C$, for all $A \in \mathcal{F}_s, s<t$,  so that  $ E[Z \mathbbm{1}_A (S_t-S_s)] =0$, which precisely means  $E_Q[ S_t \mid \mathcal{F}_s] =S_s$. Conversely, if $Q$ is a martingale probability for $S$, with bounded density $Z$,  then $$S^*_T \in L^1(P) \subseteq L^1(Q).  $$
     As $S^*_T$ is $Q$-integrable and  $\xi$ is bounded, the integral  $\xi \cdot S$      has  maximal functional $(\xi \cdot S )^*_T \in L^1(Q)$, and is thus  a martingale of class $\mathcal{H}^1(Q)$, see \cite[Chapter IV, Sect 4]{protter}). Now, if $K\in \C$ by definition   it can be super replicated at zero cost: $K\leq \xi \cdot S_T$ for some $\xi$, whence
   $$  E[Z K]= E_Q[ K ] \leq E_Q[ \xi \cdot S_T]=0.  $$
   The above inequality implies $Z\in \C^0$.
  \end{proof}

\begin{theorem}\label{ftap-simple} The following conditions are equivalent:
\begin{itemize}
  \item[a)]the market is  $\lambda$ gain-loss free,
  \item[b)] there exists an (equivalent) martingale probability $Q$  such that
    \begin{equation}\label{bound}
    \frac{\esssup \frac{dQ}{dP}}{\essinf \frac{dQ}{dP}} < \lambda.
    \end{equation}
\end{itemize}
In case any of the two conditions above holds, the market best gain-loss $\alpha^*$ admits a  dual representation as
\begin{equation}\label{dual-alpha}
    \alpha^* = \min_{Q\in \mathcal{M}_{\infty} }  \frac{\esssup \frac{dQ}{dP}}{\essinf \frac{dQ}{dP}}
\end{equation}
 in which $ \mathcal{M}_{\infty}$ is the set of  equivalent martingale probabilities $Q$  with densities $Z \in \C^0_1$ which are (bounded and) bounded away from $0$, i.e.   $\{ Z \in \C^0_1 \mid  Z>c  \text{ for some } c>0 \}$.
\end{theorem}

    The equivalence will be proved   by duality methods  via the  auxiliary utility maximization problem
     $$  u_\mu:= \sup_{K \in \mathcal{K}} E[\um( K)].   $$
       The reason is  that $u_\mu<+\infty$ is equivalent to $ \alpha^*\leq \mu$. In fact,   the functional
      $ E[\um(K) ] =E[K^+ -\mu K^- ] $  is positively homogeneous  so that
     $$ u_\mu  <+\infty    \Leftrightarrow   u_\mu = 0,$$
       and the latter condition  in turn is equivalent to  $\alpha^*\leq \mu$ because $0 \in \K $.   \\
      \indent Before starting the proof, recall also that the Fenchel pointwise inequality  \eqref{fenchel} gives, for any random variable $Y $
     $$ \um(K) - K Y \leq \Vm( Y).   $$
     \begin{proof}[Proof of Theorem \ref{ftap-simple}]
      b) $\Rightarrow$ a) If there exists a $Q$ with the stated properties, its density $Z$ belongs to $\C^0_1$ by Lemma \ref{mart}.  Set $Y =    \frac{Z}{\essinf Z} \in \C^0$ .   As $ 1\leq Y\leq  \frac{\esssup Z}{\essinf Z}:=\mu <\lambda$, $\Vm(Y)=0$ and thus for all  $K$ the Fenchel inequality simply reads as
     $  \um(K) - K Y \leq 0   $.
    Taking expectations, $  E[\um(K) ] \leq 0$ for all $K \in \K$, which is in turn equivalent to  $u_{\mu}= 0$  and to $\alpha^*\leq \mu <\lambda$.  \vspace{0.3cm}\\
     a) $\Rightarrow$ b)      Set $\mu=\alpha^*$. Then $u_\mu = 0$. The existence of a $Q$ is now  a standard duality instance.  Note that $\um$ is monotone, so $ u_\mu = \sup_{K \in \C} E[\um( K)]$. Also, the monotone concave functional $E[\um(\cdot)]$ is finite and thus continuous on $L^1$ by the Extended Namioka Theorem (see \cite{bf-namioka}, \cite{russhap2}).
       Therefore the Fenchel Duality theorem applies (see e.g. \cite[Theorem I.11 ]{bre} or  \cite{bia-eqf} for a survey of duality techniques in the utility maximization problem) and gives the formula
    $$  u_{\mu} = \min_{ Y \in \C^0} E[\Vm(Y)].  $$
    In particular the infimum  in the dual is attained  by a $Y^*  \in \C^0 $.  Therefore $ 1\leq  Y^*\leq \mu=\alpha^*<\lambda $ and its scaling   $Z^* = Y^*/E[Y^*]$ is  a martingale density with the  property  required in \eqref{bound}.\vspace{0.3cm}\\
    \indent Suppose now any of the two conditions above holds true.  Then,   the proof of the arrow b) $\Rightarrow$ a)  actually  shows \begin{equation}\label{aux}
 \alpha^*=\sup_{K\in \K, K \neq 0} \, \frac{E[K^+]}{E[K^-]} \leq \inf_{Q\in \mathcal{M}_{\infty} } \, \frac{\esssup Z}{\essinf Z},
\end{equation}
and   the proof of the arrow  a) $\Rightarrow b)$
shows that the infimum is attained by $Z^*$ and there is no duality gap.
\end{proof}

The next Corollary is essentially a slight rephrasing of the Theorem just proved.  It gives  an alternative expression for the dual representation of $\alpha^*$, which will be generalized in Corollary \ref{claim2}, Section 3.
\begin{corollary} \label{cor-maintheo}
Let $\lambda \in [1, +\infty)$ and let  $Q_{\lambda}\cap \mathcal{M}$ be the (convex) set of  martingale measures with finite $\Vl$-entropy.   The conditions: $\alpha^*<+\infty$ and $Q_{\lambda}\cap \mathcal{M} \neq \emptyset $  for some $\lambda \geq 1$ are equivalent;  and in case  $\alpha^*$ is finite, it admits the representation:
$$ \alpha^* =  \min\{ \lambda\geq 1 \mid Q_{\lambda}\cap \mathcal{M}\neq \emptyset\}$$
 In particular, $\alpha^* =1$ iff $P$ is already a martingale measure.
\end{corollary}
\begin{proof}
 Note that $\mathcal{M}_{\infty} = \cup_{\lambda \geq 1} Q_{\lambda}\cap \mathcal{M}$ and    $(Q_{\lambda}\cap \mathcal{M} )_{\lambda \geq 1}$ is a parametric family non-decreasing in $\lambda$ with $ \mathcal{Q}_1 \cap \mathcal{M}= \{P\}\cap \mathcal{M} $ either empty or equal to $  \{P\}$. The rest of the proof is then  a straightforward consequence of (the proof of) Theorem \ref{ftap-simple}.
\end{proof}

\subsection{Pros and cons of gain-loss ratio}\label{bs-unbounded}

The requirement of gain-loss free market  can thus  be seen as  a  result \`{a}-la  Fundamental Theorem of Asset Pricing also in  general probability spaces.    A comprehensive survey of  No-Arbitrage concepts and results is   the  reference book by Delbaen and Schachermayer \cite{ds-book}.   Compared to those   theorems,   the above proof looks surprisingly easy.  Of course, there is  a (twofold) reason. First, there is an integrability condition on $S$; secondly, and most importantly,  the assumption of  $\lambda$ gain-loss free market  is much  stronger than absence of arbitrage (or absence of free lunch with vanishing risk).  \\
\indent  The stronger requirement of absence of $\lambda$ gain-loss arbitrage allows a straightforward reformulation    in terms of a  standard   utility maximization  problem.  This  reformulation as such is not possible for the general FTAP case.  The reader is  however referred to \cite{rog} for a proof of the FTAP in discrete time based on a technique which relies in part on the ideas of utility maximization.

   In discrete time  trading  there is a full characterization of the models which have finite best gain-loss ratio.   On one side,   the Dalang-Morton-Willinger Theorem ensures that under No Arbitrage condition there always exists a bounded pricing kernel. Such a kernel is  not necessarily bounded away from $0$. On the other side, the  characterization of arbitrage free  markets which admit  pricing kernels satisfying prescribed lower  bounds is provided by  \cite{rs}.  \\
   \indent   In continuous time there is no such a characterization, and   $\alpha^* $ is very likely to be infinite in common models,  see  Example \ref{bs-model} for an illustration in the Black-Scholes model.
And even if it is finite,     the supremum may not be attained.  This is not due    to our specific assumptions, i.e. restriction to simple strategies in $\Xi$. In general the market best gain-loss is  intrinsically not attained, due to  the nature of the functional considered.  As it is scale invariant, maximizing sequences can be selected  without loss of generality of unitary  $L^1$-norm.  But  the unit sphere  in $L^1$ is  not (weakly) compact, unless $L^1$ is finite dimensional or, equivalently, unless $\Omega$ is finite.  So, when $\Omega$ is infinite  maximizing sequences may fail to converge, as shown in     Example \ref{sup-max}  in a one period market. \\
 \indent  Of course,  an enlargement of strategies would certainly help in capturing optimizers in some  specific model. But  given the intrinsic problems  of  gain-loss optimization,  in the end we choose to work  with simple, bounded strategies, as they have a clear financial meaning  and allow for a plain mathematical treatment.

\subsection{Examples}
\begin{example}\label{cauchy} \emph{A model where no gain  has well-defined gain-loss ratio.}  When Assumption \ref{ass-S} does not hold, gain-loss ratio criterion may lose significance.  Suppose $S$ consists of  only of one  jump  which occurs at time $T$. So, $S_t= 0$ up to time $T-$, while $S_T$ has the  distribution of the jump size.  If the filtration is the natural one, then a strategy is simply a real constant $\xi=c$ and   terminal wealths $K$ are of the form $ K = c S_T$.  Suppose the jump  has a symmetric distribution with infinite first moment. Although this is an arbitrage free model,   if $c\neq 0$  both $E[K^+]$ and $E[K^-]$ are infinite.
\end{example}

\begin{example}\label{bs-model} \emph{Gain-loss ratio is infinite in a Black-Scholes world}.
In the  Black-Scholes market model,  the density of the unique pricing kernel  is
   $$  Z= (Z_T=) \exp( - \pi W_T - \frac{\pi^2 T}{2})  $$
in which $W_T$ stands for the Brownian motion at terminal date $T$ and $ \pi = \frac{\mu -r}{\sigma}$ is the market price of risk. This density is  both unbounded  and not bounded away from $0$,  so $\C^0 $ is trivial and its basis empty. Therefore, though there is no arbitrage   when $\mu \neq r$  the Black Scholes market is not gain-loss free, for any level $\lambda$:  $\alpha^* = +\infty$.   \\
\indent  Not surprisingly, the idea behind  the construction of  explicit  arbitrarily large gain-loss ratios is playing with  sets where the density $Z$ is either  very small or very large.  The former sets have a low cost if compared to the physical probability of happening,  while the latter  in turn happen with small probability   but have a (comparatively) high cost.    We give examples of both. Without loss of generality, suppose $r=0$ and   fix  $1> \epsilon >0$.  Let $A_{\eps} : = \{ Z <\eps\}$, $p_\eps$ its probability and $X_\eps = \mathbbm{1}_{A_\eps}$, while $B_{\eps} : = \{ Z >\frac{1}{\eps}\}$,  $q_\eps$ its probability and $Y_\eps = \mathbbm{1}_{B_\eps}$.
Some calculations show that $X_\eps$ and $Y_\eps$ are cash-or-nothing digital options on $S_T=S_0 e^{(\mu- \frac{1}{2}\sigma^2)T + \sigma W_T}$, either  of call type with very large strike or of put type with very small strike when $\eps$ goes to zero.
\begin{enumerate}
  \item Let    $c_\eps = E[Z X_\eps]  $ be the cost of $X_\eps$, which is much smaller than $p_\eps$ as  $c_\eps < \eps p_\eps<1$. Since the market is complete $ K_\eps: = X_\eps- c_\eps  $ is a gain. Its gain-loss ratio is then
        $$ \frac{E[K_\eps^+ ] }{E[K_\eps^- ]} = \frac{(1-c_\eps)p_\eps}{c_\eps (1-p_\eps)} > \frac{1-c_\eps}{\eps} >\frac{1}{\eps}-p_\eps$$
        which tends to $+\infty$ as $\eps \downarrow 0$.
  \item  Let   $b_\eps = E[Z Y_\eps]  $ be the cost of $Y_\eps$. Then,     $ 1> b_\eps >  \frac{q_\eps}{\eps}$. As before, $ C_\eps: = Y_\eps- b_\eps  $  and its opposite $K_\eps$ are  gains. The  gain-loss ratio of $K_\eps$ is then
        $$ \frac{E[K_\eps^+ ] }{E[K_\eps^- ]} = \frac{ b_\eps (1-q_\eps)}{ (1-b_\eps)q_\eps} > \frac{1- q_\eps}{\eps}   $$
        which also tends to $+\infty$ as $\eps \downarrow 0$.
\end{enumerate}

The two items  together show better why  in a gain-loss free market there must be a pricing kernel bounded above \emph{and} bounded away from $0$.   As a final remark,  the  strategies that lead to the digital terminal gains $X_\eps - c_\eps$ and $Y_\eps -b_\eps$ are \emph{not} bounded.  However stochastic  integration theory, see e.g.  the book by Karatzas and Shreve \cite[Chapter 3]{ks}, ensures  they can be approximated arbitrarily well   by simple bounded strategies  with $L^2$ convergence of the terminal gains, so the approximating strategies are in $\Xi$ and their gain-loss ratio blows up.
\end{example}

\begin{example}[The market best gain-loss ratio may not be attained] \label{sup-max}  Let us consider a one period model consisting of a countable collection of one-step binomial trees, with initial uncertainty on the particular  binomial fork  we are in.  The idea is to set the odds and the (single)  risky underlying so that   the best gain-loss ratio in  the  $n$-th binomial fork  is less than the best gain-loss  in  the subsequent  $(n+1)$-th binomial fork.  This prevents the existence of an optimal solution. \\
  \indent   Suppose then $S_0 =0$, the interest rate $r=0$  and that the probability of  being in  the $n$-th fork  is $\pi_n>0$.  If we are in the $n$-th fork,  $S_1$ can either go up to a constant $c>0$, independent of $n$, or go down to $-(1+\frac{1}{n})$, with  conditional probability of going up $p_n^u$ (and $p^d_n =1-p^u_n$ is the conditional probability of going down), as summed up in the picture below.

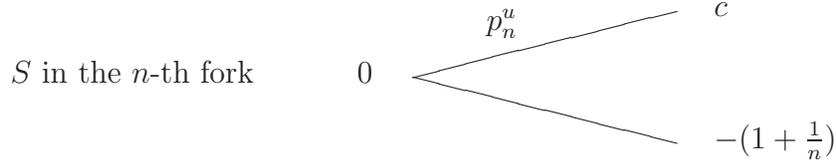
\begin{center}
\begin{picture}(50,50)
\put(26,30){\line(4,-1){100}}
\put(26,30){\line(4,1){100}}
\put(140,53){$c$}
\put(140,3){$-(1+\frac{1}{n})$}
\put(50,48){ $ p_n^u  $}
\put(-130,28){ $S$  in  the $n$-th fork }
\put(5,28){$0$}
\end{picture}
\end{center}

Since $S$ is bounded,   Assumption \ref{ass-S} is satisfied;  there is no arbitrage and $\mathcal{M}_{\infty}\neq 0$. In fact, the probability  $Q$  which gives to each fork the same probability as $P$ and gives to $S$  a conditional probability of going up in the $n$-th fork equal to $ q_n^u= \frac{1+1/n }{c+1+1/n}$ is a martingale probability which has density bounded and bounded away from $0$.     Note that  a strategy $\xi$ can be identified with the sequence  $ (\xi_n)_n $ of its  values, chosen at the beginning  of each fork.  Now,  the scale invariance property implies the best gain-loss ratio $\alpha^*_n$ in each fork is given by the best between  a long position in the underlying and a short one:

$$     \alpha^*_n = \max\left ( \frac{c p^u_n}{ (1+1/n) p^d_n},  \frac{ (1+1/n) p^d_n}{ c p^u_n} \right ). $$

If in addition  the parameters $(p^u_n)_{n\geq 1},c$  satisfy
$       \alpha^*_n < \alpha^*_{n+1}$,
then actively trading in the  $n+1$-th fork only, and do nothing in the other forks,  is always better than trading in the first $n$ forks.    To fix the ideas, suppose that in each fork being long in $S$ is better than being short, i.e.  $ \alpha^*_n = \frac{c p^u_n}{ (1+1/n) p^d_n} $. This is satisfied iff $c  \geq (1+1/n) \frac{p^d_n}{p^u_n}$ for all $n\geq 1$. Then, the condition $ \alpha^*_n  < \alpha^*_{n+1}$, for all $n$, becomes
$$  1- \frac{1}
{(n+1)^2} <\frac{p^d_n p^u_{n+1}}{p^u_n p^d_{n+1}}.$$
A simple case when this is verified is when the conditional historical probabilities do not depend on $n$.  So, suppose from now on  that $p^u_n = p^u$ for all $n$ and that $c\geq 2 \frac{p^d}{p^u}$. Then,
\begin{equation}\label{lilli}
 \alpha^*  = \lim_{n\rightarrow + \infty} \alpha^*_n = c  \frac{p^u}{ p^d}
\end{equation}
and for any  strategy $\xi $ such that $ K = \xi \cdot S_1 \in L^1$
$$  \alpha (K) <  \alpha^*$$

This is intuitive from the construction, but can be  verified by (a bit tedious and thus omitted) explicit computations with series.

 As the strategies with integrable terminal gain form the largest conceivable  domain  in gain-loss ratio maximization,  this example also proves that  the  best gain-loss ratio is  intrinsically  not attained. Namely,  it is not a matter of strategy restrictions (boundedness  or other). \\
 \indent From an analytic point of view, let us see what goes wrong.   Define  the sequence of strategies  $\xi^n$:
 $$\xi^n   = \left \{ \begin{array}{cc}
    1 & \text{ if  we are initially  in the $n$-th fork}     \\
    0 & \text{ otherwise.}
                    \end{array} \right. $$
$\xi^n$ is the optimizer in the $n$-th fork, and \eqref{lilli} implies  it is a maximizing sequence for $\alpha^*$.  The maximizing gains $k^n = \xi^n \cdot S_1$ converge in $L^1$ to $0$, but in $0$ $\alpha$ is not defined.  By   scale invariance,     the normalized version: $$  K^n = \frac{k_n}{E[ |k_n |]} $$
is still maximizing, but  is not uniformly integrable and thus has no limit.\\
   \indent We finally remark that   a $Q \in \mathcal{M}_{\infty}$  in our model exists because  the ratio of the upper value to the lower value of $S_1$ in each fork,  $(S_1)_n^u /(S_1)_n^d$,  remains bounded  and bounded away from zero when $n $ tends to infinity.   A simple modification, with e.g. $(S_1)^u_n= 1$ and $(S_1)^d_n =-2^{-n}$ as in \cite[Remark 6.5.2]{ds-book},  leads to an arbitrage free market model with no  $Q$ bounded away from zero.
\end{example}

\section{Best gain loss  with a random endowment}
\subsection{The best gain-loss $\alpha^*(B)$ is an acceptability index on $L^1$}
 Suppose the investor at time $T$ has  a non-replicable random endowment  $B\in L^1, B \notin \K$. If she   optimizes over the market in order to reduce her exposure, the best gain-loss in the presence of $B$ will  be  $$\sup_{K\in \K } \alpha(B+K),$$
 which is well defined as $B+K $ never vanishes on $ \K$. This expression  can be re-written as $ \sup_{K\in \K , K+B \neq 0} \alpha(B+K) $, which makes sense also  if $B=0$ or, more generally, if $B\in \K$, and  in that case it coincides with $\alpha^*$. From now on, the value $\alpha^*$ defined in Section 2.1  is denoted by $\alpha^*(0)$.   So, let us  define on  $L^1$ the map
  $$ \alpha^*(B):=  \sup_{K\in \K, B+K \neq 0 } \alpha(B+K).  $$

\begin{lemma} \label{Bineq}
The map  $\alpha^*$ satisfies:
 \begin{enumerate}
  \item  $\alpha^*: L^1 \rightarrow [\alpha^*(0), +\infty] $;
  \item  non-decreasing monotonicity;
  \item quasi concavity, i.e. for any $B_1, B_2 \in L^1$ and for any $c \in [0,1]$:
          \begin{equation} \label{qc}
          \alpha^*(c B_1 +(1-c)B_2) \geq \min(\alpha^*(B_1), \alpha^*(B_2))
          \end{equation}
   \item   scale invariance:
  $ \alpha^*(B)=\alpha^*( cB)  \ \ \forall c>0$
   \item continuity from below, i.e.
        $$ B_n \uparrow B \Rightarrow  \alpha^*(B_n) \uparrow \alpha^*( B). $$
 \end{enumerate}
\end{lemma}
\begin{proof}
\begin{enumerate}
  \item  Without loss of generality,  assume $B\notin \K$ and fix  $K \neq 0$. For any $t>0$, $t K \in \K$ and   by the scale invariance property of $\alpha$:
$$  \alpha ( B + tK)= \alpha ( \frac{B}{t} + K). $$
An application of dominated convergence gives   $ \lim_{t \uparrow +\infty } \alpha \left( \frac{B}{t} + K \right) \rightarrow \alpha(K)$ and consequently   $\sup_{t>0}  \alpha(\frac{B}{t}  +K) \geq \alpha(K)$.   So,
$$\alpha^*(B) = \! \sup_{K \in \K} \alpha(B+ K) =\! \sup_{K, t>0}  \alpha(B+ tK ) = \!\sup_K \left ( \!\sup_{t>0}  \alpha(\frac{B}{t}  +K) \!\right )\!\geq \! \sup_{K\neq 0} \alpha(K) = \!\alpha^*(0).$$

  \item Non-decreasing monotonicity is  a consequence of the monotonicity of $\alpha$.
  \item Quasi concavity is equivalent to  convexity of  the upper level sets  $A_b : = \{ B \in L^1 \mid \alpha^*(B)>b\}$  for any fixed $b>\alpha^*(0) = \min_B \alpha^*(B)$.  Pick $B_1, B_2 \in A_b$.  By Corollary \ref{cor-maintheo}, $ \alpha^*(0)\geq 1$, and since $b> \alpha^*(0)\geq 1$  we can assume that any maximizing sequence $K^i_n$ for $\alpha^*(B_i), i=1,2$ satisfies $ \alpha(B_i + K^i_n)>1$, or, equivalently, $B_i + K^i_n$ has positive expectation for all $n\geq 0$ and $i=1,2$.   It can be easily checked that    $\alpha$ is quasi concave when restricted to variables with positive expectation (we refer  to \cite{cm} for a proof).  Therefore, for any fixed  $c\in [0,1]$, if $W_n: = c B_1 +(1-c)B_2 + c K^1_n +(1-c)K^2_n$ we have
         $$  \alpha (W_n)  \geq \min ( \alpha(  B_1 + K^1_n), \alpha(  B_2 + K^2_n) ) $$
       and
         $\alpha^*( c B_1 + (1-c)B_2 ) \geq  \alpha( W_n)$ for all $n$.  Letting    $n\rightarrow +\infty$,   $\alpha^*( c B_1 + (1-c)B_2 )\geq \min ( \alpha^*(  B_1 ), \alpha^*(  B_2)  )>b$ and thus $c B_1 + (1-c)B_2  \in A_b $.
  \item The scale invariance property easily follows from the scale invariance of $\alpha$ and the cone property of $\K$.
   \item  Suppose  $B_n \uparrow B$. Select a maximizing sequence $(K_m)_m \in \K$ for $\alpha^*(B)$:
       $$  \alpha(B+K_m) \uparrow \alpha^*(B).$$
         For any fixed $m$, $ B_n +K_m \uparrow B +K_m$ and continuity from below of the expectation of positive and negative part implies the existence of  $n_m$ such that  $ \alpha(B_{n_m} +K_m) \geq  \alpha(B+K_m) - \frac{1}{m} $.  By the monotonicity property of $\alpha^*$:
       $$ \alpha^*(B) \geq \lim_n \alpha^*(B_n) \geq \alpha^*(B_{n_m})\geq  \alpha(B_{n_m} +K_m) \geq  \alpha(B+K_m) - \frac{1}{m} $$
       and, passing to the limit on $m$,  we get  $\alpha^*(B)  = \lim_n \alpha^*(B_n)  $.
\end{enumerate}

\end{proof}

The above lemma shows  that $\alpha^*$ is an acceptability index  continuous from below,   in the sense of Biagini and Bion-Nadal \cite{bb}.   Acceptability indexes were  axiomatically introduced by  Cherny and Madan \cite{cm}, as   maps $\beta$ defined on bounded variables with the properties:
\begin{enumerate}
  \item non-negativity
  \item non-decreasing monotonicity
  \item  quasi concavity
  \item  scale invariance
  \item  continuity from above: $B_n \downarrow B \Rightarrow \beta(B_n) \downarrow \beta (B).$
\end{enumerate}
Biagini and Bion-Nadal   extend  the analysis  of performance measures {beyond} bounded variables and in a dynamic context. In particular, here  the \emph{continuity from below}  property replaces  continuity from above.
  This   non-trivial point  is the key to the extension of the concept of acceptability indexes beyond bounded variables and  solves the value-at $0$ puzzle for indexes. In fact, continuity from above for an index,  which is $+\infty$-valued on  positive random variables (as the gain-loss ratio $\alpha$ and the optimized $\alpha^*$) implies the index should be $+\infty$-valued also at $0$. This is awkward for any index, but in particular    the best  gain-loss  index  $\alpha^*$ loses meaning if we redefine  it to be $+\infty$ at $0$ only for the sake of the (wrong) continuity requirement.

\subsection{The dual representation of $\alpha^*(B)$}
There is a  natural generalization of the results in   Theorem \ref{ftap-simple} in the presence of a claim. First, we need  an  auxiliary result.
\begin{lemma}\label{max-seq}
Fix   $B\in L^1$ and suppose $\alpha^*(B)>\alpha^*(0)$.  Then,  any maximizing sequence  $(K_n)_n$ for $\alpha^*(B)$ is bounded in $L^1$.
\end{lemma}
\begin{proof}
 Select a maximizing sequence for $\alpha^*(B)$,
$ K_n \in \K,  \alpha(B + K_n) \uparrow  \alpha^*(B)$.
Let $(c_n)_n$ denote the corresponding sequence of $L^1$-
 norms, i.e. $c_n = E[|K_n|]$. If $(c_n)_n$ were unbounded, by passing to  a subsequence, still denoted in the same way,  we could assume $c_n \uparrow +\infty$. Let $ k_n = \frac{K_n}{c_n} $. The scale invariance property of $\alpha$ would imply
$$  \alpha(B + K_n) = \frac{E[(B+ K_n)^+]}{E[(B+K_n)^-]}  =  \frac{E[(  \frac{B}{c_n}+ k_n)^+]}{E[( \frac{B}{c_n}+ k_n)^-]}    $$
Since $ \frac{B}{c_n} \rightarrow 0$ in $L^1$, then    $  \alpha^*(B)= \lim_n \alpha(B + K_n) = \lim_n \frac{E[k_n^+]}{E[k_n^-]}$, whence we would get the contradiction $\alpha^*(B)  \leq \alpha^*(0) $.
\end{proof}

\begin{theorem} \label{claim}
The following conditions are equivalent:
        \begin{itemize}
          \item[i)] $\alpha^*(B)<+\infty$
          \item[ii)]  $E_Q[B]\leq 0$ for some $Q \in \mathcal{M}_{\infty}$.
        \end{itemize}
        If any of the two conditions i), ii)   is satisfied, $\alpha^*$ admits the dual representation
\begin{equation}\label{alpha*}
\alpha^*(B) = \min_{Q \in \mathcal{M}_{\infty}, E_Q[B] \leq 0} \frac{\esssup Z}{\essinf Z},
\end{equation}
which becomes
\begin{equation}\label{alpha*strict}
\alpha^*(B) = \min_{Q \in \mathcal{M}_{\infty}, E_Q[B] = 0} \frac{\esssup Z}{\essinf Z}
\end{equation}
when  $+\infty > \alpha^*(B)>\alpha^*(0)$.
\end{theorem}

\begin{proof}[Proof \hspace{0.5cm}]

      \begin{itemize}

      \item[i)$\Rightarrow$  ii)]   Set $b= \alpha^*(B)$. Then $ b \geq \alpha^*(0)\geq 1$.   So,
        $$  0= \alpha^*(B)-b = \sup_{K\in \K } \frac{E[U_b (B+K)]}{E[(B+K)^-]}.  $$
        The denominator is positive, whence the above relation implies
       $  E[U_b (B+K)] \leq 0 $ for all $K$. Therefore $ \sup_K E[U_b (B+K)] \leq 0 $, with possibly strict inequality.  Since this supremum is finite,
           the Fenchel Duality Theorem applies, similarly to Theorem \ref{ftap-simple}, and gives:
        $$ \sup_K E[U_b (B+K)] =\min_{Q \in \mathcal{C}^0_1, y \geq 0}  \{ yE[ \frac{dQ}{dP} B] + E[ V_b ( y \frac{dQ}{dP})] \} \leq 0.$$
      Given the structure of $V_b$,   any couple of minimizers $y^*, Q^*$ satisfies $y^*>0$ and $dQ^*=Z^*dP \in  \mathcal{Q}_b \cap   \mathcal{C}^0_1  =  \mathcal{Q}_b  \cap \mathcal{M} \subseteq \mathcal{M}_{\infty}$,  which is then not empty.     So,  $   E[ V_b ( y^* \frac{dQ^*}{dP})] + y^*E_{Q^*}[B]\leq 0 $ implies $E_{Q^*}[B]\leq 0 $ and ii) follows.
      \item[ii) $\Rightarrow$ i)]   Fix a martingale measure $dQ=Z dP$ with the stated properties,  and let $ y = \frac{1}{\essinf Z} $, $\mu =    \frac{\esssup Z}{\essinf Z} $ so that  $ 1\leq y Z \leq \mu$.   The Fenchel inequality applied to the couple $ \um, \Vm$,    on $B+K$ and   $yZ$ respectively, gives
     $$  \um(B+K) - (K+B) yZ \leq \Vm(yZ)= 0   \  \ \ \forall K \in \K.$$
    Taking expectations, $  E[\um(B+ K) ] \leq yE_Q[B]\leq 0$ for all $K$,  which implies $\alpha^*(B)\leq \mu  $.
      \end{itemize}

The duality formula \eqref{alpha*} has  been implicitly proved in the above lines. In fact,   with the same notations as in the implications  i) $\rightarrow$ ii), we have the relation
$$  \alpha^*(B) \leq \frac{\esssup Z^*}{\essinf Z^*} \leq b $$
where the first inequality follows from  the arrow ii) $\rightarrow$  i), and the second from $Q^* \in \mathcal{Q}_b $.  But since $\alpha^*(B)=b$, the   inequalities  are in fact equalities. \\
\indent To show  the representation \eqref{alpha*strict},  suppose by contradiction that there exists a $B$ such that  $+\infty > \alpha^*(B)>\alpha^*(0)$ and   the minimum in \eqref{alpha*} is attained at a  $Q^*$ with $E_{Q^*}[B]<0$.   Pick a maximizing sequence $(K_n)_n$ for $\alpha^*(B)$, which by  Lemma \ref{max-seq}  is  bounded in $L^1$-norm.  With the same  notations as of the implication i) $\Rightarrow$ ii) above,    we have the inequality:
$$ E[U_b(B+K_n)] \leq y^* E_{Q^*}[B]<0.$$
From this, dividing by $ E[(B+K_n)^-]$ and adding $b$ to both members we derive
  $$  \alpha(B+K_n) = \frac{E[(B+K_n)^+]}{E[(B+K_n)^-]} \leq  b + y^* \frac{E_{Q^*}[B]  }{E[(B +K_n)^-]}  \leq   b + y^* \frac{E_{Q^*}[B]  }{ L} < b = \alpha^*(B)    $$
 where $L$  is a  uniform upper bound for $ E[(B +K_n)^-] $.  Letting $n\uparrow +\infty$, we get the contradiction $  \alpha^*(B)=\lim_n \alpha(B+K_n)   < \alpha^*(B) $.
\end{proof}

\begin{remark}
The  representations \eqref{alpha*} and \eqref{alpha*strict} are interesting \emph{per se}. In fact,     the abstract dual representation of a  quasi concave map is known (Volle, \cite[Theorem 3.4]{volle}), but there are  few examples in which such a dual representation can be explicitly computed. \\
\indent Note also   that if the market is complete and the unique martingale measure $Q^*$  is in  $ \mathcal{M}_{\infty}$,  then $\alpha^*(B) = +\infty $ iff  $E_{Q^*}[B]>0$, and $\alpha^*(B)$ is finite (and equal to $\alpha^*(0)$) if and only if $ E_{Q^*}[B]\leq 0$.
\end{remark}

\begin{corollary}\label{claim2}
With the convention $\sup  \emptyset =  \alpha^*(0)$,  $\alpha^*$   admits  the   representation

 \begin{equation}\label{alpha**}
         \alpha^*(B) = \sup \{ \lambda \geq 1   \mid     E_Q[B]> 0 \, \,   \forall Q \in \Ql \cap \mathcal{M}\}.
 \end{equation}

\end{corollary}
\begin{proof}
 With the usual convention $\inf \emptyset = +\infty $, the  proof of Theorem \ref{claim} shows that  $$ \alpha^*(B) = \inf  \{ \lambda   \mid     E_Q[B]\leq  0 \, \,   \text{ for some } Q \in \Ql \cap \mathcal{M}\} $$
 and that  $\alpha^*(B)$ is  finite iff the infimum is a minimum. As  $\Ql \cap \mathcal{M}$ is a set of probabilities which is  non-decreasing in the parameter,   the right hand side of  the above equation is  an  interval $I$, either  $[ \alpha^*(B), +\infty)$ when $\alpha^*(B)$ is finite, or empty when $\alpha^*(B)$ is infinite.
 Since
$$   \{ \lambda\geq 1   \mid     E_Q[B] >  0 \, \,   \forall  Q \in \Ql \cap \mathcal{M}\} $$
 corresponds to the interval $ I^c \cap[1, +\infty)$, its supremum  coincides with    $\alpha^*(B)$ both in the finite and infinite cases.

\end{proof}

\begin{remark}
A general result on acceptability indexes and performance measures is that any such map can be represented in  terms of a one-parameter, non-decreasing family of risk measures  (see \cite{cm,bb}).  In \cite[Theorem 1, Proposition 4]{cm} it is shown  that the gain-loss index  $\alpha$ admits a representation in terms of the family  $(\rho_{\lambda})_{\lambda} $:
$$\rho_{\lambda}(X)  := \sup_{Q\in \Ql } E_Q[-X] $$
The formula \eqref{alpha**} proves   an intuitive fact:  the market optimized  gain-loss index  $\alpha^*$  admits a representation    via the  risk measures  $(\rho_{\lambda}^M)_{\lambda}$ induced by $ (\Ql \cap \mathcal{M})_{\lambda\geq 1}$
$$  \rho_{\lambda}^M(X) := \sup_{\Ql \cap \mathcal{M}} E_Q[-X]$$
where we adopt the convention $\rho_{\lambda}^M =-\infty $ if $\Ql \cap \mathcal{M} = \emptyset $.  The family  $(\rho_{\lambda}^M)_{\lambda}$
consists of the so-called  market  modifications of the collection of risk measures $\rho_{\lambda}(X)  := \sup_{\Ql } E_Q[-X] $. For the concept of market modified risk measure and its relation with hedging, the reader is referred to  \cite{cgm} and \cite[Section 3.1.3]{bar}.
\end{remark}

\subsection{Final comments}    The results just found constitute the basis for a strong objection against best gain-loss ratio as a performance criterion in the presence of an endowment.    To start with,   Lemma \ref{Bineq}   shows that
  possessing a claim whatsoever can never be worse than the case  $B=0$ since $\alpha^*(B)\geq \alpha^*(0)$, which does not make economic sense.  \\
\indent  Second, by  Theorem \ref{claim}  the index $ \alpha^*$ can be of little use in discriminating  payoffs, as   $\alpha^*(B)$ is finite if and only if the claim belongs to $\cup_{Q\in \mathcal{M}_{\infty}} \{ B  \mid E_Q[B] \leq 0 \} $ and we have seen  that $ \mathcal{M}_{\infty}$ is empty in most continuous time models.  \\
\indent Moreover, if  there is a unique pricing kernel, say $P$, then   $\alpha^*(B)=+\infty $  if $E[B]>0$  or if $E[B]<0$ it is optimal to take infinite risk  so to off-set the negative expectation of $B$ and end up with $\alpha^*(B)=\alpha^*(0)=1$, along the same lines of   the proof of item 1 in Lemma \ref{Bineq}.  This is also unreasonable.
 \\
  \indent From a strict mathematical viewpoint,  there is quite a difference from what happens in standard utility maximization.  For example, there   if $P$ is  a martingale measure and $B=m$ is constant, the optimal solution is simply not to invest in the market. This is due to risk aversion and mathematically it is a consequence of Jensen's inequality:
$$ E[U(m + K)]\leq U(m +E[K])= U(m). $$
  On the contrary,  when $m<0$,  $0=\alpha(m) < \alpha^*(m)=1=\alpha^*(0)$.  {The scale invariance property $ \alpha^*(B) = \alpha^*(cB)  $ for all $c>0$ implies
  $$ \alpha^*(B) =   \sup_{c>0} \alpha^*(cB) = \sup_{c >0, K \in \K} \alpha( K +cB). $$
  As a consequence, our optimization problem better compares with the  so-called static/dynamic utility maximization, see e.g. Ilhan {\em et al.} \cite{ijs}, where the optimization is made dynamically in the underlyings and statically in the claim:
  $$  u(B) := \sup_{c >0, K \in \K} E[U( K + cB)] $$
  where  only long positions  are permitted in the claim  so to mirror  the constraint we have for gain-loss.
   When $P$ is a martingale measure and  $B=m<0$  the   value of  the static-dynamic utility maximization   verifies
   $$U(m)< u(m)= u(0)=U(0),$$
  and this  result is exactly  in  the spirit of the equality    $\alpha^*(m)=\alpha^*(0)$ found before.}\\
  \indent  As a final remark, the scale invariance property may be questionable for performance measures in general.  In fact, $\alpha^*$ can be seen as an evaluation  of the whole  half ray generated by $B$, $cB, c>0$, rather than $B$ itself. So, it is desirable only if the (large) investor seeks an information on the ``direction of trade'', as illustrated by Cherny and Madan \cite{cm}, and it is not appropriate for small investors, e.g.  if \emph{quantity} matters.   The cited work  \cite{bb} is entirely dedicated to the definition of a good notion of performance measures,  in an intertemporal   setting.

 \end{document}